\documentclass[10pt]{llncs}

\usepackage{ucs}
\usepackage[utf8]{inputenc}
\usepackage{amsmath}
\usepackage{amsfonts}

\usepackage{mathtools}
\usepackage{amsthm}
\usepackage{stmaryrd}
\usepackage{tikz}
\usetikzlibrary{calc}
\usepackage{xcolor}
\usepackage{graphicx}
\usepackage{multirow}
\usepackage{changepage}
\usepackage{todonotes}
\usepackage{comment}
\usepackage{blkarray}
\usepackage[symbol]{footmisc}
\usepackage[tableposition=top]{caption}

\title{The Maximum Matrix Contraction problem}
\author{
	Dimitri Watel\inst{1,2}
\and
	Pierre-Louis Poirion\inst{1,3}}
\institute{
 CEDRIC-CNAM, 292 rue du faubourg Saint Martin, 75003, Paris, FRANCE
\and
 ENSIIE, 1 Square de la résistance, Evry, FRANCE
  \email{dimitri.watel@ensiie.fr, }
\and
ENSTA Paristech
  \email{pierre-louis.poirion@ensta-paristech.fr}
}

\newcommand{\gridsize}{0.5}
\newcommand{\prgrid}[3]{\draw[step=\gridsize,gray,very thin] (#1) grid (\gridsize*#3,\gridsize*#2);}
\newcommand{\prtline}[3]{}
\newcommand{\prtcolumn}[3]{}
\newcommand{\prvtline}[3]{\draw[very thick] ($(#1)+(0,\gridsize*#2)$) -- ($(#1)+(\gridsize*#3,\gridsize*#2)$);}
\newcommand{\prvtcolumn}[3]{\draw[very thick] ($(#1)+(\gridsize*#2,0)$) -- ($(#1)+(\gridsize*#2,\gridsize*#3)$);}
\newcommand{\prvtdline}[3]{\draw[very thick,dotted] ($(#1)+(0,\gridsize*#2)$) -- ($(#1)+(\gridsize*#3,\gridsize*#2)$);}
\newcommand{\prvtdcolumn}[3]{\draw[very thick,dotted] ($(#1)+(\gridsize*#2,0)$) -- ($(#1)+(\gridsize*#2,\gridsize*#3)$);}
\newcommand{\prone}[3]{\draw ($(#1)+(#3*\gridsize-0.5*\gridsize,#2*\gridsize-0.5*\gridsize)$) node {$1$};}
\newcommand{\prbul}[3]{\draw ($(#1)+(#3*\gridsize-0.5*\gridsize,#2*\gridsize-0.5*\gridsize)$) node {$\bullet$};}

\begin{document}

\theoremstyle{plain}
\newtheorem{corol}{Corollary}

\maketitle

\begin{abstract}
In this paper, we introduce the {\it Maximum Matrix Contraction problem}, where we aim to contract as much as possible a binary matrix in order to maximize its density.  We study the complexity and the polynomial approximability of the problem. Especially, we prove this problem to be NP-Complete and that every algorithm solving this problem is at most a $2\sqrt{n}$-approximation algorithm where $n$ is the number of ones in the matrix. We then focus on efficient algorithms to solve the problem: an integer linear program and three heuristics.
\end{abstract}

\section{Introduction}
\label{sect:intro}

In this paper, we are given a two dimensional array in which some entries contain a dot and others are empty. Two lines $i$ and $i+1$ of the grid can be contracted by shifting up every dot of line $i+1$ and of every line after. Two columns $j$ and $j+1$ of the grid can be contracted by shifting left the corresponding dots. However, such a contraction is not allowed if two dots are brought into the same entry. The purpose is maximize the number of neighbor pairs of dots (including the diagonal ones). An illustration is given in Figure~\ref{fig:introduction:example}.

\begin{figure}[!ht]
	\center
  \scalebox{0.7}{
	\begin{tikzpicture}
		\coordinate (O) at (0,0);
		
		\clip (-1,-1.25) rectangle (2.5,2.1);
		
		\prgrid{O}{4}{4}
		
		\prvtdline{O}{1}{4};
		
		\prvtdcolumn{O}{1}{4};
		\prvtdcolumn{O}{2}{4};
		\prvtdcolumn{O}{3}{4};
		
		\prbul{O}{1}{2}
		\prbul{O}{1}{4}
		\prbul{O}{2}{3}
		\prbul{O}{3}{1}
		\prbul{O}{3}{3}
		\prbul{O}{4}{1}

		\draw ($(O)+(-0,0.25)$) node[anchor=east] {$4$};
		\draw ($(O)+(-0,0.75)$) node[anchor=east] {$3$};
		\draw ($(O)+(-0,1.25)$) node[anchor=east] {$2$};
		\draw ($(O)+(-0,1.75)$) node[anchor=east] {$1$};
		
		\draw ($(O)+(0.25,-0.3)$) node {$1$};
		\draw ($(O)+(0.75,-0.3)$) node {$2$};
		\draw ($(O)+(1.25,-0.3)$) node {$3$};
		\draw ($(O)+(1.75,-0.3)$) node {$4$};
		
		\draw ($(O)+(1,-0.8)$) node {$(a)$};
		
	\end{tikzpicture}
	\begin{tikzpicture}
	\coordinate (O) at (0,0);
	
	\clip (-1,-1.25) rectangle (2.5,2.1);
	
	\prgrid{O}{3}{4}
	
	\prbul{O}{1}{2}
	\prbul{O}{1}{3}
	\prbul{O}{1}{4}
	\prbul{O}{2}{1}
	\prbul{O}{2}{3}
	\prbul{O}{3}{1}
		
	\prvtdcolumn{O}{1}{3};

	\draw ($(O)+(-0,0.25)$) node[anchor=east] {$3/4$};
	\draw ($(O)+(-0,0.75)$) node[anchor=east] {$2$};
	\draw ($(O)+(-0,1.25)$) node[anchor=east] {$1$};
	
	\draw ($(O)+(0.25,-0.3)$) node {$1$};
	\draw ($(O)+(0.75,-0.3)$) node {$2$};
	\draw ($(O)+(1.25,-0.3)$) node {$3$};
	\draw ($(O)+(1.75,-0.3)$) node {$4$};
	\draw ($(O)+(1,-0.8)$) node {$(b)$};
	\end{tikzpicture}
	\begin{tikzpicture}
	\coordinate (O) at (0,0);
	
	\clip (-1,-1.25) rectangle (2.5,2.1);
	
	\prgrid{O}{3}{3}
	
	\prbul{O}{1}{1}
	\prbul{O}{1}{2}
	\prbul{O}{1}{3}
	\prbul{O}{2}{1}
	\prbul{O}{2}{2}
	\prbul{O}{3}{1}

	\draw ($(O)+(-0,0.25)$) node[anchor=east] {$3/4$};
	\draw ($(O)+(-0,0.75)$) node[anchor=east] {$2$};
	\draw ($(O)+(-0,1.25)$) node[anchor=east] {$1$};
	
	\draw ($(O)+(0.25,-0.3)$) node {$1/2$};
	\draw ($(O)+(0.75,-0.3)$) node {$3$};
	\draw ($(O)+(1.25,-0.3)$) node {$4$};
	\draw ($(O)+(0.75,-0.8)$) node {$(c)$};
	
	\end{tikzpicture}
  }
	\caption{In Figure~\ref{fig:introduction:example}.a, we give a $4 \times 4$ grid containing 6 dots. Valid contractions are represented by dotted lines and columns. It is not allowed to contract lines 1 and 2 because the two dots (1;1) and (2;1) would be brought into the same entry. Figure~\ref{fig:introduction:example}.b is the result of the contraction of lines 3 and 4 and Figure~\ref{fig:introduction:example}.c is the contraction of columns 1 and 2. The number of neighbor pairs in each grid is respectively 4, 7 and 10.
  }
	\label{fig:introduction:example}
\end{figure}
\vspace{-0.3cm}

\paragraph{Motivation. }
This problem has an application in optimal sizing of wind-farms \cite{Pillai2015} where we must first define, from a given set of wind-farms location, the neighborhood graph of this set, i.e. the graph such that two wind farms are connected if and only if their corresponding entries in the grid are neighbors. More precisely, given a set of points in the plane, we consider a first grid-embedding such that any two points are at least separated by one vertical line and one horizontal line. Then, we consider the problem of deciding which lines and columns to contract such that the derived embedding maximize the density of the grid, i.e., the number of edges in the corresponding neighbor graph.
\vspace{-0.2cm}
\paragraph{Contributions. } In this paper, we formally define the grid contraction problem as a binary matrix contraction problem in which every dot is a 1 and every other entry is 0. We study the complexity and the polynomial approximability of the problem. Especially, we prove this problem to be NP-Complete. Nonetheless, every algorithm solving this problem is at most a $2\sqrt{n}$-approximation algorithm where $n$ is the number of $1$ in the matrix. We then focus on efficient algorithms to solve the problem. We first investigate the mathematical programming formulation of MMC. We give two formulations: a straightforward non-linear program and a linear program.
Secondly, we describe three polynomial heuristics for the problem. We finally give numerical tests to compare the performances of the linear program and each algorithm. 

In Section~\ref{sec:problemdef}, we give a formal definition of the problem. In Section~\ref{sect:complexity}, we prove that the corresponding decision problem is NP-complete, then we give, in Section~\ref{sect:approx} some results about approximability of the problem. In Section~\ref{sec:linearprog} we derive a linear integer program for the model and run some experiments, then in the next section, we present and compare the three different heuristics.

\section{Problem definition}
\label{sec:problemdef}

The following definitions formalize the problem we want to solve with binary matrices. A binary matrix is a matrix with entries from $\{0,1\}$. Such a matrix modelizes a grid in which each dot is a 1 in the matrix.

\begin{definition}
Let $M$ be a binary matrix with $p$ lines and $q$ columns. For each $i \in \llbracket 1; p-1 \rrbracket$\footnote[1]{The meaning of $\llbracket p; q \rrbracket$ is the list $[p,p+1,\dots, q]$.} and each $j \in \llbracket 1; q-1 \rrbracket$, we define the \emph{line contraction matrix} $L_i$ and the \emph{column contraction matrix} $C_j$ by 
\begin{center}
\scalebox{0.7}{\mbox{$L_i = \hspace{0.2cm}\begin{blockarray}{ccccccccccccc}
1 & 2 & & & i & & & & & p && \\
\begin{block}{(cccccccccc)ccc}
1      &  0     & \cdots & 0      & 0 & 0 &  0  & 0 & \cdots & 0 &&& 1 \\
0      &  1     & \cdots & 0      & 0 & 0 &  0  & 0 & \cdots & 0 &&& 2\\
\vdots & \vdots & \ddots & \vdots & \vdots & \vdots &  \vdots  & \vdots & \ddots  & \vdots  &&&\\
0      &   0    &\cdots  & 1      & 0 & 0 &  0  & 0 & \cdots & 0 &&&\\ 
0      &   0    &\cdots  & 0      & 1 & 1 &  0  & 0 & \cdots & 0 & && i\\ 
0      &   0    &\cdots  & 0      & 0 & 0 &  1  & 0 & \cdots & 0 &&&\\
0      &   0    &\cdots  & 0      & 0 & 0 &  0  & 1 & \cdots & 0 &&&\\ 
\vdots      &   \vdots    &\ddots  & \vdots     & \vdots & \vdots & \vdots & \vdots & \ddots & \vdots  &&&\\
0      &   0    &\cdots  & 0      & 0 & 0  &  0 & 0 & \cdots & 1 &&& \\
0      &   0    &\cdots  & 0      & 0 & 0      &   0    & 0      &\cdots  & 0 &&& p\\
\end{block}
\end{blockarray}
\hspace{0.5cm}
C_j = \hspace{0.2cm}\begin{blockarray}{ccccccccccccc}
1 & 2 & & & j & & & & & q &&\\
\begin{block}{(cccccccccc)ccc}
1      &  0     & \cdots & 0      & 0 & 0 &  0  & \cdots & 0 & 0 &&& 1 \\
0      &  1     & \cdots & 0      & 0 & 0 &  0  & \cdots & 0 & 0 &&& 2\\
\vdots & \vdots & \ddots & \vdots & \vdots & \vdots &  \vdots  & \ddots & \vdots  & \vdots  &&&\\
0      &   0    &\cdots  & 1      & 0 & 0 &  0  & \cdots & 0 & 0 &&&\\ 
0      &   0    &\cdots  & 0      & 1 & 0 &  0  & \cdots & 0 & 0 &&& j\\ 
0      &   0    &\cdots  & 0      & 1 & 0 &  0  & \cdots & 0 & 0 &&&\\
0      &   0    &\cdots  & 0      & 0 & 1 &  0  & \cdots & 0 & 0 &&&\\ 
0      &   0    &\cdots  & 0      & 0 & 0  &  1 & \cdots & 0 & 0 &&&\\
\vdots      &   \vdots    &\ddots  & \vdots     & \vdots & \vdots & \vdots & \ddots & \vdots & \vdots  &&&\\
0      &   0    &\cdots  & 0      & 0 & 0      &   0    &\cdots  & 1      & 0 &&& q \\
\end{block}
\end{blockarray}.
$}}
\end{center}

The size of $L_i$ is $p \times p$ and the size of $C_j$ is $q \times q$.
\end{definition}

\begin{definition}
Let $M$ be a binary matrix of size $p \times q$, $I = [i_1, i_2, \dots, i_{|I|}]$ a sublist of $\llbracket 1;p-1 \rrbracket$ and $J = [j_1, j_2, \dots, j_{|I|}]$ a sublist of $\llbracket 1;q-1 \rrbracket$. We assume $I$ and $J$ are sorted. We define the \emph{contraction $C(M,I,J)$ of the lines $I$ and the columns $J$ of $M$} by the following matrix
$$
C(M,I,J) = \left(\prod\limits_{k = 1}^{|I|} L_{i_k}\right) \cdot M \cdot \left(\prod\limits_{k = |J|}^{1} C_{j_k}\right).
$$
\end{definition}

\begin{example}
	Let $M$ be the matrix corresponding to the grid of Figure~\ref{fig:introduction:example}.a. The following contraction gives the grid \ref{fig:introduction:example}.c:
	
	$$
	C(M,[3],[1]) = L_3 \cdot M \cdot C_1 = \begin{pmatrix}
	1 & 0 & 0 & 0 \\
	0 & 1 & 0 & 0 \\
	0 & 0 & 1 & 1 \\
	0 & 0 & 0 & 0 \\
	\end{pmatrix} \cdot \begin{pmatrix}
	1 & 0 & 0 & 0 \\
	1 & 0 & 1 & 0 \\
	0 & 0 & 1 & 0 \\
	0 & 1 & 0 & 1\\
	\end{pmatrix} \cdot \begin{pmatrix}
	1 & 0 & 0 & 0 \\
	1 & 0 & 0 & 0 \\
	0 & 1 & 0 & 0 \\
	0 & 0 & 1 & 0 \\
	\end{pmatrix} = \begin{pmatrix}
	1 & 0 & 0 & 0 \\
	1 & 1 & 0 & 0 \\
	1 & 1 & 1 & 0 \\
	0 & 0 & 0 & 0 \\
	\end{pmatrix}
	$$
\end{example}

\begin{definition}
	A contraction $C(M,I,J)$ is said \emph{valid} if and only if $C(M,I,J)$ is a binary matrix.
\end{definition}

\begin{example}
	The following contraction is not valid :
	
	$$
	C(M,[],[1,2]) = M \cdot C_2 \cdot C_1 = \begin{pmatrix}
	1 & 0 & 0 & 0 \\
	1 & 0 & 1 & 0 \\
	0 & 0 & 1 & 0 \\
	0 & 1 & 0 & 1 \\
	\end{pmatrix} \cdot \begin{pmatrix}
	1 & 0 & 0 & 0 \\
	0 & 1 & 0 & 0 \\
	0 & 1 & 0 & 0 \\
	0 & 0 & 1 & 0 \\
	\end{pmatrix} \cdot\begin{pmatrix}
	1 & 0 & 0 & 0 \\
	1 & 0 & 0 & 0 \\
	0 & 1 & 0 & 0 \\
	0 & 0 & 1 & 0 \\
	\end{pmatrix}  = \begin{pmatrix}
	1 & 0 & 0 & 0 \\
	2 & 0 & 0 & 0 \\
	1 & 0 & 0 & 0 \\
	1 & 1 & 0 & 0 \\
	\end{pmatrix}
	$$
\end{example}

\begin{definition}
	\label{def:density}
	Let $M$ be a binary matrix of size $p \times q$. The \emph{density} is the number of neighbor pairs of 1 in the matrix (including the diagonal pairs). This value may be computed with the following formula :
	$$ 
	d(M) = \frac{1}{2} \cdot \sum\limits_{i,j} \left( M_{i,j} \cdot \left(\sum\limits_{\delta = -1}^1 \sum\limits_{\gamma = -1}^1  M_{i+\delta,j+\gamma}\right) - 1 \right)
	$$
	where we define that $M_{i,j}=0$ if $(i,j) \notin \llbracket 1;p-1 \rrbracket \times \llbracket 1;q-1 \rrbracket$
\end{definition}

\begin{problem}\label{problem1}
	\emph{Maximum Matrix Contraction problem} (MMC). Given a binary matrix $M$ of size $p \times q$ such that $n$ entries equal 1 and $p \cdot q - n$ entries equal 0, the Maximum Matrix Contraction problem consists in the search for two sublists $I$ of $\llbracket 1;p-1 \rrbracket$ and $J$ of $\llbracket 1;q-1 \rrbracket$ such that the contraction $C(M,I,J)$ is valid and maximizes $d(C(M,I,J))$.
\end{problem}

We study in the next two sections the complexity and the approximability of this problem.

\section{Complexity}
\label{sect:complexity}
 
This section is dedicated to proving the NP-Completeness of the problem.

  \begin{theorem}
  \label{theo:complexity}
  The decision version of (MMC) is NP-Complete.
  \end{theorem}
  \begin{proof}

    Let $M$ be an instance of MMC. Given an integer $K$, a sublist $I$ of $\llbracket 1; p-1 \rrbracket$ and a sublist $J$ of $\llbracket 1; q-1 \rrbracket$, we can compute in polynomial time the matrix $C(M,I,J)$ and check if the contraction is valid and if $d(C(M,I,J)) \geq K$. This proves the problem belongs to NP.

    In order to prove the NP-Hardness, we describe a polynomial reduction from the NP-Complete Maximum Clique problem \cite{Karp1972}. Lets $G(V,E)$ be an instance of the Maximum Clique problem, we build an instance $M$ of MMC with $p = q = (4|V| + 6)$. We arbitrarily number the nodes of $G$ : $V = \{v_1, v_2, \dots v_{|V|}\}$.

    Let $l_i$ and $c_i$ be respectively the $6+4(i-1)+1$-th line and the $6+4(i-1)+1$-th column. We associate the four lines $l_i, l_i+1, l_i+2, l_i+3$ and the four columns $c_i, c_i +1, c_i+2, c_i + 3$ to $v_i$. The key idea of the reduction is that each node $v$ is associated with two 1 of the matrix. If we choose the node $v$ to be in the clique, then, firstly, the two 1 associated with $v$ are moved next to each other and this increases the density by one; and secondly, for every node $w$ such that $(v,w) \not\in E$, the two 1 associated to cannot be moved anymore.

A complete example is given in Figure~\ref{fig:reduction:example}. For each node $v_i$, we set $M_{l_i,c_i} = M_{l_i+2,c_i+2} = 1$. If the nodes $v_i$ and $v_j$ are not linked with an edge, we set $M_{l_i,c_j} = M_{l_i+1,c_j+1} = 1$. If, on the contrary, there is an edge $(v_i,v_j)$, then the intersections of the lines of $v_i$ and the column of $v_j$ is filled with 0. Finally, we add some 1 in the six first columns and the six first lines of the matrix such that only the contractions of the line $l_i$ and the column $c_i$ for $i \in \llbracket 1;n \rrbracket$ are valid.


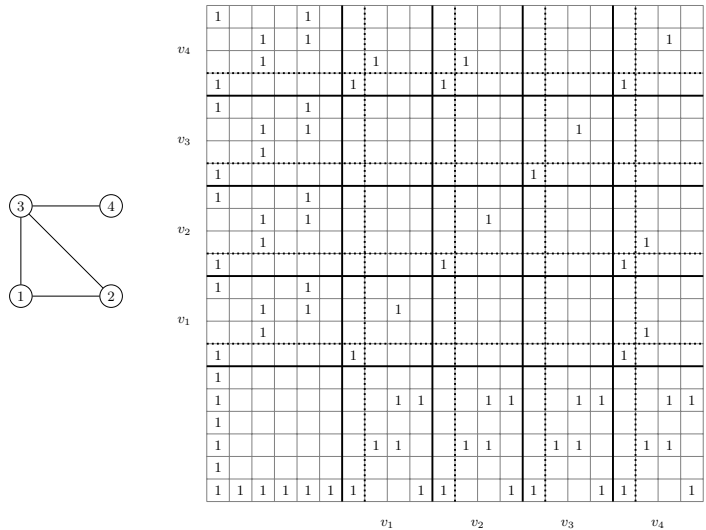
\begin{figure}
\centering
\scalebox{0.6}{
		\begin{tikzpicture}
		
		\clip (-0.25,-5.25) rectangle (3.25,7.25);
		\tikzset{tinoeud/.style={draw, circle, minimum height=0.5cm}}
		\node[tinoeud] (U) at (0,0) {};
		\node[tinoeud] (V) at (2,0) {};
		\node[tinoeud] (W) at (0,2) {};
		\node[tinoeud] (T) at (2,2) {};

    \draw (U) -- (V);
    \draw (V) -- (W);
    \draw (W) -- (U);
    \draw (W) -- (T);

    \draw (U) node {$1$};
    \draw (V) node {$2$};
    \draw (W) node {$3$};
    \draw (T) node {$4$};

		\end{tikzpicture}

		\begin{tikzpicture}
		
		\coordinate (O) at (0,0);

    \prgrid{O}{22}{22}

    \prvtline{O}{6}{22}
    \prvtline{O}{10}{22}
    \prvtline{O}{14}{22}
    \prvtline{O}{18}{22}
    
    \prvtcolumn{O}{6}{22}
    \prvtcolumn{O}{10}{22}
    \prvtcolumn{O}{14}{22}
    \prvtcolumn{O}{18}{22}
    
    \prvtdline{O}{7}{22}
    \prvtdline{O}{11}{22}
    \prvtdline{O}{15}{22}
    \prvtdline{O}{19}{22}
    
    \prvtdcolumn{O}{7}{22}
    \prvtdcolumn{O}{11}{22}
    \prvtdcolumn{O}{15}{22}
    \prvtdcolumn{O}{19}{22}
    
    \draw ($(O)+(-0.5,4)$) node {$v_1$};
    \draw ($(O)+(-0.5,6)$) node {$v_2$};
    \draw ($(O)+(-0.5,8)$) node {$v_3$};
    \draw ($(O)+(-0.5,10)$) node {$v_4$};
    
    \draw ($(O)+(4,-0.5)$) node {$v_1$};
    \draw ($(O)+(6,-0.5)$) node {$v_2$};
    \draw ($(O)+(8,-0.5)$) node {$v_3$};
    \draw ($(O)+(10,-0.5)$) node {$v_4$};

    \prone{O}{1}{1};
    \prone{O}{1}{2};
    \prone{O}{1}{3};
    \prone{O}{1}{4};
    \prone{O}{1}{5};
    \prone{O}{1}{6};
    
    \prone{O}{2}{1};
    \prone{O}{3}{1};
    \prone{O}{4}{1};
    \prone{O}{5}{1};
    \prone{O}{6}{1};
 
    \prone{O}{7}{1};
    \prone{O}{8}{3};
    \prone{O}{9}{3};
    \prone{O}{9}{5};
    \prone{O}{10}{5};
    \prone{O}{10}{1};
 
    \prone{O}{11}{1};
    \prone{O}{12}{3};
    \prone{O}{13}{3};
    \prone{O}{13}{5};
    \prone{O}{14}{5};
    \prone{O}{14}{1};
 
    \prone{O}{15}{1};
    \prone{O}{16}{3};
    \prone{O}{17}{3};
    \prone{O}{17}{5};
    \prone{O}{18}{5};
    \prone{O}{18}{1};
 
    \prone{O}{19}{1};
    \prone{O}{20}{3};
    \prone{O}{21}{3};
    \prone{O}{21}{5};
    \prone{O}{22}{5};
    \prone{O}{22}{1};

    \prone{O}{1}{7};
    \prone{O}{3}{8};
    \prone{O}{3}{9};
    \prone{O}{5}{9};
    \prone{O}{5}{10};
    \prone{O}{1}{10};
 
    \prone{O}{1}{11};
    \prone{O}{3}{12};
    \prone{O}{3}{13};
    \prone{O}{5}{13};
    \prone{O}{5}{14};
    \prone{O}{1}{14};
 
    \prone{O}{1}{15};
    \prone{O}{3}{16};
    \prone{O}{3}{17};
    \prone{O}{5}{17};
    \prone{O}{5}{18};
    \prone{O}{1}{18};
 
    \prone{O}{1}{19};
    \prone{O}{3}{20};
    \prone{O}{3}{21};
    \prone{O}{5}{21};
    \prone{O}{5}{22};
    \prone{O}{1}{22};

    \prone{O}{7}{7};
    \prone{O}{9}{9};

    \prone{O}{11}{11};
    \prone{O}{13}{13};

    \prone{O}{15}{15};
    \prone{O}{17}{17};

    \prone{O}{19}{19};
    \prone{O}{21}{21};

    \prone{O}{7}{19};
    \prone{O}{8}{20};

    \prone{O}{19}{7};
    \prone{O}{20}{8};

    \prone{O}{11}{19};
    \prone{O}{12}{20};

    \prone{O}{19}{11};
    \prone{O}{20}{12};

    \end{tikzpicture}}

\caption{This figure illustrates, on the left, a graph in which we search for a maximum clique and, on the right, the matrix obtained built with the reduction. We do not show the 0 entries of the matrix for readability. The dotted lines and columns represent the valid contractions.}
\label{fig:reduction:example}
\end{figure}

The initial density in this matrix is $d_0 = 11 + 6|V| + (|V| (|V|-1) - 2|E|)$. Note that, in order to add one to the density of the matrix, the only way is to choose a node $v_i$ and contract the column $c_i$ and the line $l_i$. If the column $c_i$ is contracted and if $(v_i, v_j) \not \in E$, the two entries $M_{l_j,c_i}$ and $M_{l_j+1,c_i+1}$ are moved on the same column. Similarly, if the line $l_i$ is contracted, the two entries $M_{l_i,c_j}$ and $M_{l_i+1,c_j+1}$ are moved on the same line. This prohibits the contraction of the line $l_j$ and the column $c_j$. Consequently, in order to add $C$ to the density, we must find a clique of size $C$ in the graph and contract every line and column associated with the nodes of that clique.

Thus, there is a clique of size $K$ if and only if there is a feasible solution for $M$ of density $d_0 + K$. This concludes the proof of NP-Completeness.
\end{proof}
  
The Maximum Clique problem cannot be approximated to within $|V|^{\frac{1}{2} - \varepsilon}$ in polynomial time unless P = NP \cite{Hastad1999}. Unfortunately, the previous reduction cannot be used to prove a negative approximability result occurs for MMC. Indeed, the density of any feasible solution of the MMC instance we produce is between $d_0 + 1$ and $d_0 + |V|$, with $d_0 = O(|V|^2 - |E|)$. Consequently, the optimal density is at most $(1 + 1/|V|)$ times the worst density. A way to prove a higher inapproximability ratio for MMC would be to modify the reduction such that the gap between the optimal solution and another feasible solution increases.

In the next section, we prove that a $n^{\frac{1}{2} - \varepsilon}$ harness ratio would almost tight the approximability of MMC as there exists a $2\sqrt{n}$-approximation algorithm.

\section{Approximability}
\label{sect:approx}

In this section we define the notion of maximal feasible solution and prove that every algorithm returning a maximal feasible solution is a $2\sqrt{n}$-approximation where $n$ is the number of 1 in the matrix.

\begin{definition}
We say a feasible solution is \emph{maximal} if it is not strictly included in another feasible solution. In other words, when all the lines and columns of that solution are contracted, no contraction of any other line or column is valid.
\end{definition}

\begin{lemma}
\label{lem:bounds}
Let $M$ be an instance of MMC, $(I,J)$ be a maximal feasible solution and $M' = C(M,I,J)$ then $2 \sqrt{n} \leq d(M') \leq 4n$.
\end{lemma}
\begin{proof}
A 1 in a matrix cannot have more than $8$ neighbors, thus the density of $M'$ is no more than $4n$.

Let $p'$ and $q'$ be respectively $p - |I|$ and $q - |J|$. For each line $i \in \llbracket 1,p'-1 \rrbracket$ of $M'$, there is a column $j$ such that $M'_{i,j} = M'_{i+1,j} = 1$, otherwise we could contract line $i$ and $(I,J)$ would not be maximal. Similarly for each column $j \in \llbracket 1,q'-1 \rrbracket$. Thus $d(M') \geq p'+q'$ where $p' \times q'$ is the size of $M'$. From the inequality of arithmetic and geometric means, we have $p' + q ' \geq 2 \sqrt{p'\cdot q'}$ and, as $M'$ contains $n$ entries such that $M'_{i,j} = 1$, $p'\cdot q' \geq n$. Thus $p' + q ' \geq 2 \sqrt{n}$.
\end{proof}

From the upper bound and the lower bound given in the previous lemma, we can immediately prove the following theorem. 

\begin{theorem}
	\label{theo:sqrtnapprox}
An algorithm returning any maximal solution of an instance of MMC is a $2\sqrt{n}$-approximation.
\end{theorem}

Theorem~\ref{theo:sqrtnapprox} proves a default ratio for every algorithm trying to solve the problem. Note that there are instances in which the ratio between an optimal density and the lowest density of a maximal solution is $O(\sqrt{n})$. An example is given in the external report in \cite{WatelPoirionAppendix}. In Section~\ref{sect:heuristics}, we describe three natural heuristics to solve the problem. We show in \cite{WatelPoirionAppendix} that their approximability ratio is $O(\sqrt{n})$ by exhibiting a worst case instance.

Determining if MMC can be approximated to within a constant factor is an open question. As it was already pointed at the end of Section~\ref{sect:complexity}, the problem may possibly be not approximable to within $n^{\frac{1}{2}-\varepsilon}$ and this would almost tight the approximability of MMC.

The next two sections focus on efficient algorithms to solve the problem. The next section is dedicated to the mathematical programming methods.

\section{Linear integer programming}
\label{sec:linearprog}

For $i \in \llbracket 1; p-1 \rrbracket$ (resp.  $j \in \llbracket 1; q-1 \rrbracket$ ), let $x_i$ (resp. $y_j$) be the binary variable such that $x_i=1$ (resp. $y_j=1$) if and only if line $i$ is contracted, i.e. $i \in I$ (resp. column $j$ is contracted, i.e. $j \in J$). From the definitions of Section~\ref{sec:problemdef}, we can model the MMC problem by the following non-linear binary program:

\begin{equation*}
(\ast)\left\{
\begin{array}{lll}
\max\limits_{x,y}  \quad	& d(A)  \\
& A= \prod\limits_{i=1}^{p-1}((L_i-I_p)x_i+I_p)M\prod\limits_{j=q-1}^{1}((C_j-I_q)y_j+I_q) \\
& A_{i,j} \le 1, \quad \forall (i,j) \in \llbracket 1; p-1 \rrbracket \times \llbracket 1; q-1 \rrbracket \\
& x_i,y_j \in \{0,1\}
\end{array}\right.
\end{equation*}
where $I_p$ denotes the identity matrix of size $p$ and where the formula of $d(A)$ is the one given in Definition~\ref{def:density}.

\noindent Although this formulation is very convenient to write the mathematical model, it is intractable as we would need to add an exponential number of linearizations: for all subset $ I,J \subseteq \llbracket 1; p-1 \rrbracket \times \llbracket 1; q-1 \rrbracket$ we would need a variable $x_I=\prod\limits_{i \in I}x_i $ and $y_I=\prod\limits_{j \in J}y_j $.\\

\noindent We now present a linear integer programming model for the MMC problem: instead of linearizing the products $\prod\limits_{i \in I}x_i$ and $\prod\limits_{j \in J}y_j$, we cut the product \\
$A= \prod\limits_{i=1}^{p-1}((L_i-I_p)x_i+I_p)M\prod\limits_{j=q-1}^{1}((C_j-I_q)y_j+I_q) $ in $T=p+q-1$ time-steps. More precisely, define $A^{(1)}=M$; for all $t=2,...,p$, we define by $A^{(t)}$ the matrix which is computed after deciding the value of $y_j$ for $j \ge p-t+1$; similarly, for all $p+1\le t \le T$, $A^{(t)}$ is determined by the value of $y_j$ for all $j$ and by the value of $x_i$ for $i \ge q -t+p $. We obtain the following program:

\begin{equation*}
(P)\left\{
\begin{array}{lll}
\max\limits_{x,y}  \quad	& d(A^{(T)})  \\
& A^{(t+1)}= ((L_{p-t}-I_p)x_{p-t}+I_p)A^{(t)} \qquad & \forall 1\le t\le p-1\\
& A^{(t+1)}= A^{(t)}((C_{q-t+p}-I_q)y_{q-t+p}+I_q) \qquad & \forall p\leq t\le T\\
& A^{(t)}_{i,j} \le 1, \quad \forall (i,j,t) \in \llbracket 1; p-1 \rrbracket \times \llbracket 1; q-1 \rrbracket \times \llbracket 2; T \rrbracket \\
& x_i,y_j \in \{0,1\}
\end{array}\right.
\end{equation*}

\noindent We can easily linearize the model above by introducing, for all $(i,j,t) \in \llbracket 1; p-1 \rrbracket \times \llbracket 1; q-1 \rrbracket \times \llbracket 2; T \rrbracket $ $r_{i,j,t}=A^{(t)}_{i,j}*x_{p-t}$ if $1\le t\le p-1$ and $r_{i,j,t}=A^{(t)}_{i,j}*y_{q-t+p}$ if $p+1\le t\le T$, noticing that the variables $A^{(t)}_{i,j},x_t,y_t$ are all binary. Finally, after linearizing the product $A^{(T)}_{i,j}A^{(T)}_{k,l}$ in the objective function, $d(A^{(T)})$, we obtain a polynomial size integer programming formulation of the MMC problem.

\subsection{Numerical results}

We test the proposed model using IBM ILOG CPLEX 12.6. The experiments are performed on an Intel i7 CPU at 2.70GHz with 16.0 GB RAM. The models are implemented in Julia using JuMP \cite{Lubin2015}. The algorithm is run on random squared matrices. Given a size $p$ and a probability $r$, we produce a random binary matrix $M$ of size $p \times p$ such that $Pr(M_{i,j} = 1) = r$. The expected value of $n$ is then $r \cdot p^2$. We test the model for $n \in \{6,9,12\}$ for a probability $r \in \{0.1,0.15,0.2,0.25,0.3\}$ and we report the optimal value $d^*$ and the running time. For each value of $p$ and $r$, 10 random instances are created, whose averages are reported on Table~\ref{tab:lp}.

\vspace{-0.5cm}

\begin{table}[ht!]
	\centering
	\caption{Test of random instances for the integer linear program model.}
	\def\arraystretch{1.2}
	\setlength\tabcolsep{0.075cm}
	\small
	\begin{tabular}{| c | c | c | c | c | c |c| c| c| c| c| }
		\hline
		& \multicolumn{10}{c |}{r} \\
		\hline
		& \multicolumn{2}{c |}{0.1} & \multicolumn{2}{c |}{0.15} & \multicolumn{2}{c |}{0.20} & \multicolumn{2}{c |}{0.25} & \multicolumn{2}{c |}{0.3} \\
		\hline
		& $d^*$ & time (s) & $d^*$ & time (s) & $d^*$ & time (s) & $d^*$ & time (s) & $d^*$ & time (s) \\
		\hline
		(p,q)=(6,6) & 6.0 & 0.3 & 4.1 & 0.26 & 12.1 & 0.2 & 15.3 & 0.28 & 22.0 & 0.15 \\ 
		\hline
		(p,q)=(9,9) & 15.1 & 5.3 & 22.1 & 5.1 & 32.3 & 7.8 & 36.5 & 7.0 & 44.5 & 3.4 \\ 
		\hline
		(p,q)=(12,12) & 30.8 & 171.6 & 48.0 & 281.2 & 55.0 & 183.0 & 64.4 & 101.0 & 71.0 & 95.1 \\ 
		\hline
	\end{tabular}
	\label{tab:lp}
\end{table}

We notice that the integer programming model is not very efficient to solve the problem. For $p=15$, in most of the cases,  CPLEX needs to run more than 2 hours to solve the model.

\section{Polynomial heuristics}
\label{sect:heuristics}

In this section, we describe three heuristics for MMC : a first-come-first-served algorithm and two greedy algorithms.

\subsection{The LCL heuristic}

This algorithm is a first-come-first-served algorithm. It is divided into two parts: the Line-Column (LC) part and the Column-Line (CL) part. 

The LC part computes and returns a maximal feasible solution $M^{LC}$ by, firstly, contracting a maximal set of lines $I^{LC}$ and, then, by contracting a maximal set of columns $J^{LC}$. The algorithm builds $I^{LC}$ as follows: it checks for each line from $p-1$ down to $1$ if the contraction of that line is valid. In that case, the contraction is done and the algorithm goes on. $J^{LC}$ is built the same way.

The CL part computes and returns a maximal feasible solution $M^{CL}$ by starting with the columns and ending with the lines. The LCL algorithm then returns the solution with the maximum density.

The advantage of such an algorithm is its small time complexity.

\begin{theorem}
	The time complexity of the LCL algorithm is $O(p \cdot q)$. 
\end{theorem}
\begin{proof}
	The four sets $I^{LC}$, $J^{LC}$, $I^{CL}$ and $J^{CL}$ can be implemented in time $O(p \cdot q)$ using an auxiliary matrix $M'$. The proof is given for the first one, the implementation of the three other ones is similar. At first, we copy $M$ into $M'$. For each line $i$ from $p-1$ to $1$ of $M'$, we check with $2q$ comparisons if there is a column $j$ such that $M'_{i,j} = M'_{i+1,j} = 1$. In that case, we do nothing. Otherwise, we add $i$ to $I^{LC}$ and we replace line $i$ with the sum of the $i$-th and the $i+1$-th lines.
	
	Finally, given a matrix $M$ and a set of lines $I$, one can compute $C(M,I,\emptyset)$ in time $O(p \cdot q)$ by, firstly, computing in time $O(p)$ an array $A$ of size $p$ such that $A_i$ is the number of lines in $I$ strictly lower than $i$ and, secondly, returning a matrix $C$ of size $p - |I| \times q$ such that $C_{i-A_i,j} = M_{i,j}$.
\end{proof}

\begin{remark}
	Note that, if there is at most one 1 per line of the matrix of the matrix, the LCL algorithm is asymptotically a 4-approximation when $n$ approaches infinity. Indeed, the LC part returns a line matrix in which each entry is a 1. The density of this solution is $n-1$. As the maximum density is $4n$ by Lemma~\ref{lem:bounds}, the ratio is $4\frac{n}{n-1}$. On the contrary, an example given in the external report \cite{WatelPoirionAppendix} proves that this algorithm is, in the worst case, at least a $O(\sqrt{n})$-approximation. 
\end{remark} 

\subsection{The greedy algorithm}

The greedy algorithm tries to maximize the density at each iteration. The algorithm computes $d(C(M,\{i\},\emptyset))$ and $d(C(M,\emptyset, \{j\}))$ for each line $i$ and each column $j$ if the contraction is valid. It then chooses the line or the column maximizing the density. It starts again until the solution is maximal.

\begin{theorem}
	The time complexity of the Greedy algorithm is $O(p^2 \cdot q^2)$. 
\end{theorem}
\begin{proof}
	There are at most $p \cdot q$ iterations. At each iteration, we compute one density per line $i$ and one density per column $j$. The density of $C(M,\{i\},\emptyset)$ is the density of $M$ plus the number of new neighbor pairs of 1 due to the contraction of lines $i$ and $i+1$. The increment can be computed in time $O(q)$ as there are at most three new neighbors for each of the $q$ entries of the four lines $i-1$ to $i+2$. Similarly, the density of $C(M,\emptyset,\{j\})$ can be computed in time $O(p)$. Thus one iteration takes $O(p \cdot q)$ iterations. 
\end{proof}

\begin{remark}
	We prove in \cite{WatelPoirionAppendix} that the greedy algorithm is at least a $O(\sqrt{n})$-approximation algorithm.
\end{remark}

\subsection{The neighborization algorithm}

The neighborization algorithm is a greedy algorithm trying to maximize, at each iteration, the number of couple of entries that can be moved next to each other with a contraction. This algorithm is designed to avoid the traps in which the LCL algorithm and the Greedy algorithm fall in by never contracting lines and columns that could prevent some 1 entries to gain a neighbor.

We define a function $N$ from $(\llbracket 0;p-1 \rrbracket \times \llbracket 0;q-1 \rrbracket)^2$ to $\{0,1\}$.
For each couple $c = ((i,j),(i',j'))$ such that $M_{i,j} = 0$ or $M_{i',j'} = 0$, $N(c) = 0$. Otherwise, $N(c) = 1$ if and only if there is a sublist of lines $I$ and a sublist of columns $J$ such that $C(M,I,J)$ is valid and such that the two entries are moved next to each other with this contraction. Finally, we define $N(M)$ as the sum of all the values $N((i,j),(i',j'))$. The algorithm computes $N(C(M,\{i\},\emptyset))$ and $N(C(M,\emptyset, \{j\}))$ for each line $i$ and each column $j$ if the contraction is valid. It chooses the line or the column maximizing the result and starts again until the solution is maximal.

\begin{theorem}
	The time complexity of the Greedy algorithm is $O(n^2 \cdot p^3 \cdot q^3 \cdot (p+q))$. 
\end{theorem}
\begin{proof}
	Let $M$ be a binary matrix, we first determine the time complexity we need to compute $N(M)$. Let $((i,j),(i',j'))$ be two coordinates such that $M_{i,j} = M_{i',j'} = 1$. We assume $i < i'$ and $j < j'$. The two entries may be moved next to each other if $i'- i -1$ of the $i'-i$ lines and $j'- j -1$ of the $j'-j$ columns between the two entries may be contracted and this can be done in time $O(p \cdot q \cdot (j'-j) \cdot (i'-i)) = O(p^2 \cdot q^2)$. As there are at most $n^2$ entries satisfying $M_{i,j} = M_{i',j'} = 1$, we need $O(n^2 \cdot p^2 \cdot q^2)$ operations to compute $N(M)$.  
	
	As there are at most $p \cdot q$ iterations. At each iteration, we computes one value per line $i$ and one value per column $j$ in time $O(n^2 \cdot p^2 \cdot q^2)$. The time complexity is then $O(n^2 \cdot p^3 \cdot q^3 \cdot (p+q))$.
\end{proof}

\begin{remark}
	We prove in \cite{WatelPoirionAppendix} that the neighborization algorithm is at least a $O(\sqrt{n})$-approximation algorithm.
\end{remark}

\subsection{Numerical results}

In this last subsection, we give numerical results of the three algorithms in order to evaluate their performances.

The experiments are performed on an Intel(R) Core(TM) i7-4810MQ CPU @ 2.80GHz processor with 8Go of RAM. The algorithms are implemented with Java 8\footnote{The implementations can be found at \url{https://github.com/mouton5000/MMCCode}.}. The algorithms are run on random squared matrices. Given a size $p$ and a probability $r$, we produce a random binary matrix $M$ of size $p \times p$ such that $Pr(M_{i,j} = 1) = r$. The expected value of $n$ is then $r \cdot p^2$. Before executing each algorithm, we first reduce the size of each instance by removing every column and line with no 1.

\paragraph{Small instances. }
We first test the three algorithms on small instances on which we can compute an exact brute-force algorithm. This algorithm exhaustively enumerates every subset of lines and columns for which the contraction is valid and returns the solution with maximum density. The results are summarized on Table~\ref{tab:heuristices:numericalresults:smallinstances:compareToExact} and  Table~\ref{tab:heuristices:numericalresults:smallinstances:compareToEachOther}.

\begin{table}[ht!]
	\centering
	\def\arraystretch{1.2}
	\setlength\tabcolsep{0.075cm}
	\scriptsize
	\caption{This table details the results for each algorithm. For each values of $p$ and $r$, the algorithms are executed on 50 instances. We give for each heuristic the mean running time in milliseconds, the mean ratio between the optimal density $d^*$ and returned density $d$ and the number of instances for which the ratio is 1.}
	\begin{tabular}{| c | c | c | c | c | c | c | c | c |  c | c | c | }
	\hline
	& & Exact & \multicolumn{3}{c|}{LCL} & \multicolumn{3}{c|}{Greedy} & \multicolumn{3}{c|}{Neigh.} \\
	\hline
	$p$ & $r$ & time (ms) & time (ms) & $\frac{d^*}{d}$ & $d = d^*$ & time (ms) &  $\frac{d^*}{d}$ & $d = d^*$ & time (ms) &  $\frac{d^*}{d}$ & $d = d^*$ \\
	\hline
	\multirow{8}{*}{5} & 0.01 & $<$ 1 ms & $<$ 1 ms & 1 & 50 & $<$ 1 ms & 1 & 50 & $<$ 1 ms & 1 & 50\\
	\cline{2-12}
	& 0.02 & $<$ 1 ms & $<$ 1 ms & 1 & 50 & $<$ 1 ms & 1 & 50 & $<$ 1 ms & 1 & 50\\
	\cline{2-12}
	& 0.03 & $<$ 1 ms & $<$ 1 ms & 1 & 50 & $<$ 1 ms & 1 & 50 & $<$ 1 ms & 1 & 50\\
	\cline{2-12}
	& 0.04 & $<$ 1 ms & $<$ 1 ms & 1.00 & 49 & $<$ 1 ms & 1 & 50 & $<$ 1 ms & 1 & 50\\
	\cline{2-12}
	& 0.05 & $<$ 1 ms & $<$ 1 ms & 1.00 & 48 & $<$ 1 ms & 1 & 50 & $<$ 1 ms & 1 & 50\\
	\cline{2-12}
	& 0.1 & $<$ 1 ms & $<$ 1 ms & 1.00 & 46 & $<$ 1 ms & 1.00 & 49 & $<$ 1 ms & 1 & 50\\
	\cline{2-12}
	& 0.2 & $<$ 1 ms & $<$ 1 ms & 1.00 & 45 & $<$ 1 ms & 1.00 & 46 & $<$ 1 ms & 1 & 50\\
	\cline{2-12}
	& 0.3 & $<$ 1 ms & $<$ 1 ms & 1.00 & 43 & $<$ 1 ms & 1.00 & 45 & 2.52 & 1.00 & 49\\
	\hline
	\multirow{8}{*}{10} & 0.01 & $<$ 1 ms & $<$ 1 ms & 1.00 & 48 & $<$ 1 ms & 1 & 50 & $<$ 1 ms & 1 & 50\\
	\cline{2-12}
	& 0.02 & $<$ 1 ms & $<$ 1 ms & 1.02 & 46 & $<$ 1 ms & 1.00 & 46 & $<$ 1 ms & 1 & 50\\
	\cline{2-12}
	& 0.03 & $<$ 1 ms & $<$ 1 ms & 1.04 & 37 & $<$ 1 ms & 1.00 & 41 & 1.22 & 1.00 & 49\\
	\cline{2-12}
	& 0.04 & $<$ 1 ms & $<$ 1 ms & 1.02 & 35 & $<$ 1 ms & 1.00 & 39 & 1.92 & 1.00 & 49\\
	\cline{2-12}
	& 0.05 & $<$ 1 ms & $<$ 1 ms & 1.10 & 28 & $<$ 1 ms & 1.00 & 26 & 1.98 & 1.00 & 46\\
	\cline{2-12}
	& 0.1 & 2.60 & $<$ 1 ms & 1.00 & 19 & $<$ 1 ms & 1.00 & 21 & 15.50 & 1.00 & 34\\
	\cline{2-12}
	& 0.2 & $<$ 1 ms & $<$ 1 ms & 1.00 & 23 & $<$ 1 ms & 1.00 & 23 & 66.42 & 1.00 & 40\\
	\cline{2-12}
	& 0.3 & $<$ 1 ms & $<$ 1 ms & 1.00 & 31 & $<$ 1 ms & 1.00 & 34 & 66.64 & 1.00 & 42\\
	\hline
	\multirow{8}{*}{15} & 0.01 & $<$ 1 ms & $<$ 1 ms & 1.16 & 33 & $<$ 1 ms & 1.00 & 43 & $<$ 1 ms & 1 & 50\\
	\cline{2-12}
	& 0.02 & $<$ 1 ms & $<$ 1 ms & 1.06 & 21 & $<$ 1 ms & 1.00 & 25 & 1.64 & 1.00 & 40\\
	\cline{2-12}
	& 0.03 & $<$ 1 ms & $<$ 1 ms & 1.08 & 17 & $<$ 1 ms & 1.00 & 17 & 4.36 & 1.00 & 40\\
	\cline{2-12}
	& 0.04 & 3.76 & $<$ 1 ms & 1.02 & 11 & $<$ 1 ms & 1.00 & 15 & 14.84 & 1.00 & 34\\
	\cline{2-12}
	& 0.05 & 9.40 & $<$ 1 ms & 1.02 & 18 & $<$ 1 ms & 1.00 & 14 & 38.96 & 1.00 & 33\\
	\cline{2-12}
	& 0.1 & 295.74 & $<$ 1 ms & 1.00 & 6 & $<$ 1 ms & 1.00 & 8 & 355.54 & 1.00 & 19\\
	\cline{2-12}
	& 0.2 & 28.24 & $<$ 1 ms & 1.00 & 14 & $<$ 1 ms & 1.00 & 18 & 892.10 & 1.00 & 33\\
	\cline{2-12}
	& 0.3 & $<$ 1 ms & $<$ 1 ms & 1.00 & 30 & $<$ 1 ms & 1.00 & 37 & 541.58 & 1.00 & 45\\
	\hline
	\multirow{8}{*}{20} & 0.01 & $<$ 1 ms & $<$ 1 ms & 1.18 & 23 & $<$ 1 ms & 1.00 & 31 & 1.04 & 1.00 & 45\\
	\cline{2-12}
	& 0.02 & 59.06 & $<$ 1 ms & 1.14 & 10 & $<$ 1 ms & 1.00 & 15 & 21.24 & 1.00 & 29\\
	\cline{2-12}
	& 0.03 & 431.60 & $<$ 1 ms & 1.04 & 9 & $<$ 1 ms & 1.00 & 6 & 119.82 & 1.00 & 20\\
	\cline{2-12}
	& 0.04 & 2275.64 & $<$ 1 ms & 1.00 & 2 & $<$ 1 ms & 1.00 & 5 & 273.82 & 1.00 & 19\\
	\cline{2-12}
	& 0.05 & 10223.92 & $<$ 1 ms & 1.00 & 3 & $<$ 1 ms & 1.00 & 4 & 622.92 & 1.00 & 8\\
	\cline{2-12}
	& 0.1 & 44268.36 & $<$ 1 ms & 1.00 & 7 & $<$ 1 ms & 1.00 & 2 & 3809.98 & 1.00 & 17\\
	\cline{2-12}
	& 0.2 & 424.84 & $<$ 1 ms & 1.00 & 15 & $<$ 1 ms & 1.00 & 11 & 5302.22 & 1.00 & 33\\
	\cline{2-12}
	& 0.3 & $<$ 1 ms & $<$ 1 ms & 1.00 & 34 & $<$ 1 ms & 1.00 & 46 & 1553.86 & 1.00 & 49\\
	\hline
\end{tabular}
	\label{tab:heuristices:numericalresults:smallinstances:compareToExact}
\end{table}

\begin{table}[ht!]
	\centering
	\def\arraystretch{1.2}
	\setlength\tabcolsep{0.075cm}
	\scriptsize
	\caption{Each entry of this table details, for each couple of heuristics, the number of instances of Table~\ref{tab:heuristices:numericalresults:smallinstances:compareToExact} (there are 1600 instances) for which the line heuristic gives a strictly better results than the column heuristic. }
	\begin{tabular}{| c | c | c | c | }
	\hline
	& LCL & Greedy & Neigh. \\
	\hline
	LCL& - & 366 & 70\\
	\hline
	Greedy & 426& - & 86\\
	\hline
	Neigh & 629 & 587& - \\
	\hline
\end{tabular}
	\label{tab:heuristices:numericalresults:smallinstances:compareToEachOther}
\end{table}

We can observe from Table~\ref{tab:heuristices:numericalresults:smallinstances:compareToExact} that the running time first increases when $r$ grows and then decreases. Similarly, the number of times the heuristics return an optimal solution first decreases and then increases. The first behavior is explained by the fact that the size of instances with small values of $n$ can be reduced. On the other hand, if $r$ is high, the number of lines and columns of which the contraction is not valid increases and, then, the search space of the algorithms is shortened. Considering the running times, as it was predicted by the time complexities, the LCL and the greedy heuristics are the fastest algorithms. We can observe that the neighborization algorithm can be slower than the exact algorithm on small instances because the running time of the former is more influenced by $n$ than the latter. However, we do not exclude the fact the implementation of the neighborization algorithm may be improved. Considering the quality of the solutions returned by the algorithms, according to Tables~\ref{tab:heuristices:numericalresults:smallinstances:compareToExact} and \ref{tab:heuristices:numericalresults:smallinstances:compareToEachOther}, the neighborization heuristic shows better performances than the greedy and the LCL algorithms.

\paragraph{Big instances. }
We then test the two fastest algorithms LCL and Greedy on bigger instances. The results are given on Table~\ref{tab:heuristices:numericalresults:biginstances}.  Four interesting differences with Table~\ref{tab:heuristices:numericalresults:smallinstances:compareToExact} emerges from Table~\ref{tab:heuristices:numericalresults:biginstances}. Firstly, the LCL algorithm is faster than the greedy algorithm. This is coherent with the time complexities. Secondly, the LCL algorithm does not follow the same behavior as the exact algorithm and the neighborization heuristic for small instances: the running time increases with $r$ even if the search space is shortened. Indeed, contrary to the three other algorithms, the implementation does not depend on this search space. Thirdly, the running time of the greedy algorithm first increases with $r$, then decreases and and finally slowly increases again. This last increase is due to the computation time of the density and the line and columns that can be contracted.  Finally, the solution returned by the LCL algorithm seems to be better for small values of $r$ and, on the other hand, the greedy algorithm returns better densities for middle values. The two algorithms are equivalent for high values of $r$ because those instances can probably not be contracted.

\begin{table}[ht!]
	\centering
	\def\arraystretch{1.2}
	\setlength\tabcolsep{0.05cm}
	\scriptsize
	\caption{This table details the results for the LCL algorithm and the greedy alorithm. For each values of $p$ and $r$, the algorithms are executed on 50 instances. We give for each heuristic the mean running time in milliseconds and how many times the returned density is strictly better than the density returned by the other algorithm.}
	\begin{tabular}{| c | c | c | c | c | c | }
	\hline
	& & \multicolumn{2}{c|}{LCL} & \multicolumn{2}{c|}{Greedy}\\
	\hline
	$p$ & $r$ & \ time (ms) & $d_{L} > d_{G}$ & time (ms) &  $d_{L} < d_{G}$ \\
	\hline
	\multirow{8}{*}{200} & 0.01 & $<$ 1 ms & 49 & 17.78 & 1\\
	\cline{2-6}
	& 0.02 & $<$ 1 ms & 48 & 22.58 & 2\\
	\cline{2-6}
	& 0.03 & $<$ 1 ms & 43 & 21.82 & 5\\
	\cline{2-6}
	& 0.04 & $<$ 1 ms & 31 & 19.26 & 18\\
	\cline{2-6}
	& 0.05 & $<$ 1 ms & 21 & 16.76 & 29\\
	\cline{2-6}
	& 0.1 & 1.28 & 10 & 5.18 & 40\\
	\cline{2-6}
	& 0.2 & 1.92 & 0 & $<$ 1 ms & 0\\
	\cline{2-6}
	& 0.3 & 2.58 & 0 & $<$ 1 ms & 0\\
	\hline
	\multirow{8}{*}{500} & 0.01 & 3.28 & 50 & 382.06 & 0\\
	\cline{2-6}
	& 0.02 & 3.58 & 44 & 321.30 & 6\\
	\cline{2-6}
	& 0.03 & 3.92 & 17 & 237.06 & 33\\
	\cline{2-6}
	& 0.04 & 4.56 & 10 & 164.82 & 40\\
	\cline{2-6}
	& 0.05 & 4.88 & 4 & 104.48 & 46\\
	\cline{2-6}
	& 0.1 & 6.80 & 0 & 4.70 & 2\\
	\cline{2-6}
	& 0.2 & 10.66 & 0 & 3.34 & 0\\
	\cline{2-6}
	& 0.3 & 15.06 & 0 & 4.58 & 0\\
	\hline
\end{tabular}\begin{tabular}{| c | c | c | c | c | c | }
		\hline
		& & \multicolumn{2}{c|}{LCL} & \multicolumn{2}{c|}{Greedy}\\
		\hline
	$p$ & $r$ & time (ms) & $d_{L} > d_{G}$ & time (ms) &  $d_{L} < d_{G}$ \\
	\hline
	\multirow{8}{*}{1000} & 0.01 & 12.00 & 50 & 2832.52 & 0\\
	\cline{2-6}
	& 0.02 & 14.04 & 21 & 1890.40 & 29\\
	\cline{2-6}
	& 0.03 & 16.34 & 1 & 1099.38 & 49\\
	\cline{2-6}
	& 0.04 & 17.72 & 1 & 553.90 & 49\\
	\cline{2-6}
	& 0.05 & 18.74 & 5 & 233.70 & 45\\
	\cline{2-6}
	& 0.1 & 24.72 & 0 & 7.82 & 0\\
	\cline{2-6}
	& 0.2 & 41.50 & 0 & 12.62 & 0\\
	\cline{2-6}
	& 0.3 & 59.18 & 0 & 18.36 & 0\\
	\hline
	\multirow{8}{*}{2000} & 0.01 & 53.54 & 49 & 22068.00 & 1\\
	\cline{2-6}
	& 0.02 & 59.96 & 0 & 10664.44 & 50\\
	\cline{2-6}
	& 0.03 & 65.66 & 0 & 3914.08 & 50\\
	\cline{2-6}
	& 0.04 & 71.68 & 6 & 1049.00 & 44\\
	\cline{2-6}
	& 0.05 & 76.36 & 0 & 186.04 & 10\\
	\cline{2-6}
	& 0.1 & 100.16 & 0 & 28.88 & 0\\
	\cline{2-6}
	& 0.2 & 167.42 & 0 & 50.46 & 0\\
	\cline{2-6}
	& 0.3 & 237.54 & 0 & 72.88 & 0\\
	\hline
\end{tabular}
	\label{tab:heuristices:numericalresults:biginstances}
\end{table}

\section{Conclusion}

In this paper, we introduced the Maximum Matrix Contraction problem (MMC). 
We proved this problem is NP-Complete. However, we also proved that every algorithm which solves this problem is an $O(\sqrt{n})$-approximation algorithm. Considering that the NP-Completeness was derived from the Maximum Clique problem, and that this problem cannot be polynomially approximated to within $n^{\frac{1}{2}-\varepsilon}$, MMC is very likely to not being approximable to within the same ratio. Such a result would almost tight the approximability of MMC.

Moreover, we studied four algorithms to solve the problem, an integer linear program, a first-come-first-served algorithm and two greedy algorithms, and gave numerical results. It appears firstly that integer linear programming is not adapted to MMC while the three other heuristics returns really good quality solutions in short amount of time even for large instances. Those results seems to disconfirm the $n^{\frac{1}{2}-\varepsilon}$ inapproximability ratio. It would be interesting to deepen the study in order to produce a constant-factor polynomial approximation algorithm or a polynomial-time approximation scheme if such an algorithm exists.

\bibliographystyle{splncs}
\bibliography{edla}

\begin{thebibliography}{1}

\bibitem{Pillai2015}
Pillai, A., Chick, J., Johanning, L., Khorasanchi, M., de~Laleu, V.:
\newblock {Offshore wind farm electrical cable layout optimization}.
\newblock Engineering Optimization \textbf{47}(12) (2015)  1689--1708

\bibitem{Karp1972}
Karp, R.M.:
\newblock {Reducibility among combinatorial problems}.
\newblock In: Complexity of Computer Computations.
\newblock Springer (1972)  85--103

\bibitem{Hastad1999}
H{\aa}stad, J.:
\newblock {Clique is hard to approximate within n{\^{}}(1-$\epsilon$)}.
\newblock Acta Mathematica \textbf{182}(1) (1999)  105--142

\bibitem{WatelPoirionAppendix}
Watel, D., Poirion, P.:
\newblock {The Maximum Matrix Contraction problem : Appendix}.
\newblock Technical Report {CEDRIC-16-3645}, {CEDRIC laboratory, CNAM, France}
  (2016)

\bibitem{Lubin2015}
Lubin, M., Dunning, I.:
\newblock {Computing in operations research using Julia}.
\newblock INFORMS Journal on Computing \textbf{27}(2) (2015)  238--248

\end{thebibliography}

\end{document}